\title{Subspace Polynomials and Cyclic Subspace Codes}
\author{Eli Ben-Sasson\IEEEauthorrefmark{2} \and Tuvi Etzion\IEEEauthorrefmark{1} \and Ariel Gabizon\IEEEauthorrefmark{2} \and Netanel Raviv\IEEEauthorrefmark{1}}
\date{}
\documentclass[10pt,journal,draftclsnofoot,onecolumn]{IEEEtran}

\usepackage[margin=0.7in]{geometry}
\usepackage{amssymb}
\usepackage{amsmath}
\usepackage{amsthm}
\usepackage{units}
\usepackage{color}
\usepackage{filecontents}
\usepackage{hyperref}

\newtheorem{theorem}{Theorem}
\newtheorem{definition}{Definition}

\newtheorem{lemma}{Lemma}
\newtheorem{corollary}{Corollary}
\newtheorem{conjecture}{Conjecture}
\newtheorem{construction}{Construction}

\newtheorem{example}{Example}

\newtheorem{remark}{Remark}

\newcommand{\cl}[1]{\mathcal{#1}}

\newcommand{\sus}{\subseteq}

\newcommand{\bC}{\mathbb{C}}

\newcommand{\bF}{\mathbb{F}}

\newcommand{\bN}{\mathbb{N}}

\newcommand{\bZ}{\mathbb{Z}}

\newcommand{\cP}{\mathcal{P}}

\newcommand{\qbin}[3]{{#1 \brack #2}_{#3}}
\newcommand{\grsmn}[3]{\cl{G}_{#1}\left(#2,#3\right)}

\DeclareMathOperator{\gap}{gap}
\DeclareMathOperator{\lcm}{lcm}
\DeclareMathOperator{\ord}{ord}

\begin{document}
\maketitle

\begin{abstract}
Subspace codes have received an increasing interest recently due to their application in
error-correction for random network coding. In particular, cyclic subspace codes are possible candidates
for large codes with efficient encoding and decoding algorithms. In this paper we consider
such cyclic codes and provide constructions of optimal codes for which their codewords do not
have full orbits. We further introduce a new way to represent subspace codes by a class of polynomials
called subspace polynomials. We present some constructions of such codes which are cyclic
and analyze their parameters.
\end{abstract}

\footnotetext[1]{This research was supported in part by the Israeli
Science Foundation (ISF), Jerusalem, Israel, under
Grant 10/12.}

\footnotetext[2]{This research was supported in part by the European Community's Seventh Framework
Programme (FP7/2007-2013) under grant agreements number 257575 and 240258.}

\footnotetext[0]{
The work of Netanel Raviv is part of his Ph.D.
thesis performed at the Technion.

The authors are with the Department of Computer Science, Technion,
Haifa 3200003, Israel.

e-mail: \{eli,etzion,arielga,netanel\}@cs.technion.ac.il.}

\section{Introduction}
Let $\bF_q$ be the finite field of size $q$, and let
$\bF_q^*\triangleq \bF_q \setminus\{0\}$. For $n\in \bN$ denote
by $\bF_{q^n}$ the field extension of degree $n$ of $\bF_q$ which may be seen as the vector space of dimension $n$ over $\bF_q$. By abuse of notation, we will not distinguish between these two concepts. Given a non-negative integer $k\le n$, the set of all $k$-dimensional
subspaces of~$\bF_{q^n}$ forms a \textit{Grassmannian} space
(Grassmannian in short) over $\bF_q$, which is denoted by $\grsmn{q}{n}{k}$. The size of $\grsmn{q}{n}{k}$ is given by the well-known Gaussian coefficient~$\qbin{n}{k}{q}$.
The set of all subspaces of~$\bF_{q^n}$ is called the
\textit{projective space} of order $n$ over $\bF_q$ \cite{EV11-ECCinPS} and
is denoted by $\cP_q(n)$. The set $\cP_q(n)$ is endowed
with the metric $d(U,V)=\dim U + \dim V-2\dim(U\cap V)$.
A subspace code is a collection $\bC$ of subspaces from $\cP_q(n)$.
In this paper we will be mainly interested in
\textit{constant dimension} codes (called also Grassmannian codes),
that is, $\bC \subseteq \grsmn{q}{n}{k}$ for some $k\le n$.

Subspace codes and constant dimension codes have attracted a lot of
research in the last eight years. The motivation was given in~\cite{KK08-CodingFor},
where it was shown how subspace codes may be used in random network
coding for correction of errors and erasures. This application
of subspace codes renewed the interest in a wide variety of problems
related to vector spaces \cite{OldRes3,OldRes1,OldRes2,OldRes4}, particularly in constructions of large
codes with error correction capability, efficient encoding algorithms
for these codes, as well as efficient decoding algorithms.

In~\cite{KK08-CodingFor} a novel construction of large
subspace codes using \textit{linearized polynomials} (a.k.a. $p$-polynomials~\cite{O33-onSpecial})
is presented. These codes were later shown~\cite{SKK08-rankMetricApproach}
to be related to optimal rank-metric codes through an operation called \textit{lifting}.
These two techniques and some of their variants
are the main known tools for constructing subspace codes.

It was previously suggested \cite{BEOVW,EV11-ECCinPS,
KoKu08}
that \textit{cyclic subspace codes} may present a useful structure that
can be applied efficiently for the purpose of coding. For a subspace
$V\in\grsmn{q}{n}{k}$ and $\alpha \in \bF_{q^n}^*$ we define the
\textit{cyclic shift} of $V$ as $\alpha V \triangleq \{\alpha v ~\vert~ v\in V\}$.
The set $\alpha V$ is clearly a subspace of the same dimension as~$V$. Two cyclic shifts are called \textit{distinct} is they form two different subspaces.
A~subspace code $\bC$ is called \textit{cyclic} if for
every $\alpha \in \bF_{q^n}^*$ and every $V\in\bC$ we have $\alpha V \in \bC$.

In~\cite{EV11-ECCinPS,KoKu08} several examples of optimal cyclic subspace codes
with small dimension were found. In~\cite{BEOVW} an optimal code
which also forms a $q$-analog of Steiner system was presented. This
code has an automorphism group which is generated by a cyclic shift
and the Frobenius mapping (known together also as a normalizer of a
Singer subgroup~\cite{BEOVW},\cite[pp. 187-188]{H67-Endliche}).
These codes raised the plausible conjecture
that large cyclic codes may be constructed in any dimension.
However, the current approaches for construction of subspace codes fall
short with handling cyclic codes. In this paper we aim at establishing
new general techniques for constructions of cyclic codes.

In~\cite{TMBR13-CyclicOrbitCodes} a thorough algebraic analysis
of the structure of cyclic orbit codes is given.
One class of such codes is the cyclic codes. However,
no nontrivial construction is given. In~\cite{GMT14-StabSub}
a construction of cyclic codes with degenerated orbit (of size less than $\frac{q^n-1}{q-1}$) is given. This construction produces a subcode of some codes in our work (see Section~\ref{section:deg_orbit}). Both~\cite{GMT14-StabSub} and~\cite{TMBR13-CyclicOrbitCodes} raised the following conjecture:

\begin{conjecture}
\label{conjecture:2kminus2}
For every positive integers $n,k$ such that $k< n/2$, there
exists a cyclic code of size~$\frac{q^n-1}{q-1}$ in $\grsmn{q}{n}{k}$ and minimum distance $2k-2$.
\end{conjecture}

Notice that for $k>n/2+1$, a minimum distance of $2k-2$ is clearly not possible. The original conjecture~\cite{TMBR13-CyclicOrbitCodes} considered $k\le n/2$. However, an exhaustive search which was used in~\cite{GMT14-StabSub} proved that the conjecture is false for $n=8,k=4,q=2$. When $k<n/2$, it appears that there is enough flexibility that many such codes exist, while for $k=n/2$, such a code might not exist. Its existence depends on the existence of a subspace which forms a structure similar to a difference set~\cite{BEOVW}. In this paper it is proved that this conjecture is true for a given $k$
and infinitely many values of $n$, along with several options for explicit
constructions (see Theorem~\ref{theorem:explicit2kminus2}).
In~\cite{GMT14-StabSub,TMBR13-CyclicOrbitCodes} it was also
pointed out that it is not known how to construct cyclic codes
with multiple orbits. In the sequel we show that our techniques
can be useful for this purpose (see Lemma~\ref{lemma:multipleOrbitUsingPrimality}
and Construction~\ref{construction:Cd}).

One of the tools in our constructions is the so-called subspace polynomials,
which are a special case of linearized polynomials. Subspace polynomials form
an efficient method of representing subspaces, from which one can directly
deduce certain properties of the subspace which are not evident
in some other representations. These objects were studied in
the past for various purposes,
e.g., bounds on list-decoding of Reed-Solomon and rank-metric codes~\cite{A13-Bounds},
construction of affine dispersers~\cite{BK12-Affine},
and finding an element of high multiplicative order in a finite field~\cite{CQ12-HighOrder}.

The rest of this paper is organized as follows. Section
II will start with the known definition of subspace
polynomials. We continue to analyze properties of the
subspaces corresponding to the subspace polynomials,
in particular we examine distance properties induced by
cyclic and Frobenius shifts of these subspaces. Based on these properties, in Section~\ref{section:constructions} we consider constructions of optimal cyclic codes with degenerate orbits, and cyclic codes with full orbits. The main goal in constructing cyclic codes is to obtain as many orbits as possible
in the code. This task will be left for future work. In this work we consider
first the existence of cyclic codes with one full length orbit and cyclic
codes with multiple full length orbits. Conclusions are given in Section~\ref{section:Conclusions}.

\section{Subspaces and their Subspace Polynomials} \label{section:Perliminaries}
For the rest of this paper $k$ and $n$ will be positive integers
such that $2 < k < n$, and we denote $[\ell]\triangleq q^\ell$. We begin by defining linearized polynomials and subspace polynomials.

\begin{definition} \label{definition:linearizedPoly} A linearized polynomial was defined by Ore \cite{O33-onSpecial} as follows:
\[
P(x)\triangleq a_k \cdot x^{[k]}+a_{k-1}\cdot x^{[k-1]}+\cdots +a_1\cdot x^{[1]}+a_0\cdot x
\]
where the coefficients are in the finite field $\bF_{q^n}$.
\end{definition}
Linearized polynomials have numerous applications in classic coding theory (e.g., \cite[Chapter~4]{MS-77}). It is widely known that the roots of any linearized polynomial form a subspace in some extension of $\bF_{q^n}$ (seen as a vector space over $\bF_q$) and for every $V\in\grsmn{q}{k}{n}$, the polynomial $\prod_{v\in V}(x-v)$ is a linearized polynomial \cite[p. 118]{MS-77}. We will be particulary interested in linearized polynomials that have simple roots with respect to some field $\bF_{q^n}$.

\begin{definition} \label{definition:subspacePoly} \cite{BK12-Affine,BK10-Limits,
CQ12-HighOrder,A13-Bounds}
A monic linearized polynomial $P$ with coefficients in $\bF_{q^n}$ is called a subspace polynomial with respect to $\bF_{q^n}$ if the following equivalent conditions hold:
\begin{enumerate}
\item P divides $x^{[n]}-x$.
\item $P$ splits completely over $\bF_{q^n}$ and all its roots have multiplicity 1.
\end{enumerate}
\end{definition}

From now on, we shall omit the notation of $\bF_{q^n}$ whenever it is clear from context. The first two lemmas are trivial and well known. The simplicity of the roots of a subspace polynomial (and in particular, the simplicity of 0) gives rise to the following lemma.

\begin{lemma} \label{lemma:a0Not0}
In any subspace polynomial, the coefficient of $x$ is non-zero. Conversely, every linearized polynomial with non-zero coefficient of $x$ is a subspace polynomial in its splitting field.
\end{lemma}

\begin{proof}
It is readily verified that 0 is a root of multiplicity 1 if and only if the coefficient of $x$ is non-zero. Therefore, if $P$ is a subspace polynomial, all of his roots are of multiplicity 1 (see Definition~\ref{definition:subspacePoly}), including 0. On the other hand, if $Q$ is a linearized polynomial with a non-zero coefficient of $x$, then by \cite[Theorem 3.50, p. 108]{LN97-FiniteFields}, all the roots of $Q$ have multiplicity~$1$.
\end{proof}

It also follows from Definition \ref{definition:subspacePoly} that for a given $V\in\grsmn{q}{n}{k}$ the polynomial $\prod_{v\in V}(x-v)$ is the unique subspace polynomial whose set of roots is $V$, which leads to the following lemma.
\begin{lemma} \label{lemma:subspaceEqual}
Two subspaces are equal if and only if their corresponding subspace polynomials are equal.
\end{lemma}

Lemma~\ref{lemma:subspaceEqual} allows us to denote by $P_V$ the unique subspace polynomial corresponding to a given subspace $V$.
\begin{example}
Let $t$ be a positive integer such that $t\vert n$. It is known that $\bF_{q^t}$ is a subfield (in particular, a subspace) of $\bF_{q^n}$. The subspace polynomial of $\bF_{q^t}$ is $P_{\bF_{q^t}}(x)=x^{[t]}-x$.
\end{example}
The connection between linearized polynomials and subspace polynomial is given by the following two claims.

\begin{theorem}\cite[Theorem 3.50, p. 108]{LN97-FiniteFields}\label{theorem:lienarizedMultiplicity}
If $P$ is a linearized polynomial whose splitting field is $\bF_{q^n}$, then each root of $P$ in~$\bF_{q^n}$ has the same multiplicity, which is a non-negative power of $q$, and the roots form a linear subspace of $\bF_{q^n}$.
\end{theorem}

\begin{lemma}
If $P(x)$ is a linearized polynomial with a leading coefficient\footnote{The leading coefficient of a polynomial is the coefficient of the monomial with the highest degree.} $\alpha\ne 0$ and the splitting field of $P(x)$ is $\bF_{q^n}$, then $P(x)=\alpha P_V(x)^{[t]}$  for some subspace $V$ in $\bF_{q^n}$ and some $t\in\bN$.
\end{lemma}

\begin{proof}
According to Theorem~\ref{theorem:lienarizedMultiplicity}, all the roots of $P$ are of the same multiplicity $q^t$ for some $t\in\bN$, and these roots form a subspace $V$ of $\bF_{q^n}$. Hence,
\begin{eqnarray*}
P(x)=\alpha\prod_{v\in V} (x-v)^{[t]}
=\alpha\left( \prod_{v\in V}(x-v)\right)^{[t]}=\alpha P_V(x)^{[t]}.
\end{eqnarray*}\end{proof}

In the sequel, we show several connections between the coefficients of subspace polynomials and properties of the respective subspaces. One of the main tools in our analysis is the difference between the indices of the two topmost non-zero coefficients:

\begin{definition}\label{definition:gap} For $V\in\grsmn{q}{n}{k}$ and $P_V(x)=x^{[k]}+\sum_{j=0}^{i}\alpha_j x^{[j]}$, where $\alpha_i\ne 0$, let $\gap(V)\triangleq k-i$.
\end{definition}

As the following lemma illustrates, the gap of two subspace induces a lower bound on their related distance.

\begin{lemma} \label{lemma:intersectionLemma}
If $V \in \grsmn{q}{n}{k_1}$ and $U \in \grsmn{q}{n}{k_2}$ are two distinct subspaces such that $k_1 \le k_2$ and
\begin{eqnarray*}
P_V (x) &=& x^{[k_1]}+\sum_{j=0}^{t}\alpha_j x^{[j]}\\ 
P_U (x) &=& x^{[k_2]}+\sum_{j=0}^{s}\beta_j x^{[j]},
\end{eqnarray*} 
such that $\alpha_t\ne 0$ and $\beta_s\ne 0$, then $\dim\left(U \cap V\right) \le \max(s,t+k_2-k_1)$.
\end{lemma}

\begin{proof}
According to the properties of $\bF_{q^n}$, for all $\alpha,\beta\in\bF_{q^n}$ and for all $i\in\bN$ we have that $(\alpha+\beta)^{[i]}=\alpha^{[i]}+\beta^{[i]}$, and therefore
\begin{eqnarray*}
P_V(x)^{[k_2-k_1]}=x^{[k_2]}+\sum_{j=0}^t \alpha_j^{[k_2-k_1]}x^{[j+k_2-k_1]}.
\end{eqnarray*}
Since the polynomials $P_V,P_V^{[k_2-k_1]}$ have the same set of roots, and since the roots of~$P_U$ are simple, it follows that\footnote{$\gcd(s,t)$ stands for the greatest common denominator of the elements $s,t$.}$\gcd(P_V,P_U)=\gcd(P_V^{[k_2-k_1]},P_U)$. Hence, 
if $Q(x) \triangleq P_U(x)-P_V(x)^{[k_2-k_1]}$ then
\begin{eqnarray*}
\gcd(P_V,P_U)&=&\gcd(P_V^{[k_2-k_1]},P_U)\\
&=&\gcd(P_V^{[k_2-k_1]},P_U(mod~P_V^{[k_2-k_1]}))\\
&=&\gcd(P_V^{[k_2-k_1]},~Q(mod~P_V^{[k_2-k_1]})).
\end{eqnarray*}
Since $\deg Q\le \max([s],[t+k_2-k_1])$, it follows that \[\log_q \deg \gcd(P_V^{[k_2-k_1]}, Q(mod~P_V^{[k_2-k_1]})) \le \max(s,t+k_2-k_1),\]and hence $\dim(U\cap V)\le \max(s,t+k_2-k_1)$.
\end{proof}

%

A special case of Lemma~\ref{lemma:intersectionLemma}, where the subspaces $U$ and $V$ are of the same dimension $k$, provides the following useful corollaries.

\begin{corollary} \label{corollary:intersectionUpperBound}
If $U,V\in\grsmn{q}{n}{k}$ then $\dim(U\cap V)\le k-\min (\gap(U),\gap(V))$.
\end{corollary}
\begin{corollary} \label{corollary:distanceLowerBound}
If $U,V\in\grsmn{q}{n}{k}$ then $d(U,V)\ge 2\min (\gap(U),\gap(V))$.
\end{corollary}

\begin{remark}\label{remark:gapNotTight}
Corollary~\ref{corollary:distanceLowerBound} is not tight, i.e., there exists subspaces $U,V\in\grsmn{q}{n}{k}$ where $\gap(V)=gap(U)=1$ and $d(U,V)=2k-2$. For example, let $\gamma$ be a root of $x^7+x+1=0$, and use this primitive polynomial to generate $\bF_{2^7}$. The following polynomials are subspace polynomials of $U,V\in\grsmn{2}{7}{3}$ for which $\gap(U)=\gap(V)=1$ and $d(U,V)=2\cdot 3-2\cdot 1=4$. In particular, $U$ and~$V$ are cyclic shifts of each other.
\begin{eqnarray*}
P_U(x) &=& x^{[3]} + x^{[2]} + (\gamma^6 + \gamma^4 + \gamma^3 + \gamma + 1)x^{[1]} + (\gamma^3 + \gamma^2 + \gamma + 1)x\\
P_V(x) &=& x^{[3]} + (\gamma^2 + 1)x^{[2]} + (\gamma^6 + \gamma^4 + \gamma + 1)x^{[1]} + (\gamma^5 + \gamma^4 + \gamma)x
\end{eqnarray*}
\end{remark}

Aside from cyclic shifts we will also use the well known \textit{Frobenius mapping} $F^i$
as a method to increase the size of the codes.
For an element $\alpha\in \bF_{q^n}$ and $i\in\{0,\ldots,n-1\}$,
the ${\bF_q\mbox{-}}$mapping $F^i$ is defined as $F^i(\alpha)=\alpha^{[i]}$
(see \cite[p. 75]{LN97-FiniteFields}). For a subspace~$V$ and
$i\in\{0,\ldots,n-1\}$ the $i$th \textit{Frobenius shift} of~$V$
is defined as ${F^i(V)\triangleq \{v^{[i]}~\vert~v\in V\}}$.
Since the function~$F^i$ is an automorphism, it follows
that the set $F^i(V)$ is a subspace of the same dimension as~$V$.
We now characterize the subspace polynomials of the subspaces resulting from these mappings.

\begin{lemma} \label{lemma:cyclicShiftPoly}
If $V\in \grsmn{q}{n}{k}$ and $\alpha \in \bF_{q^n}^{*}$ then $P_{\alpha V}(x) = \alpha ^{[k]}\cdot P_V (\alpha^{-1}x)$. That is, if $P_V(x) = x^{[k]} +\sum_{j=0}^i \alpha_j x^{[j]}$ then $P_{\alpha V}(x)=x^{[k]} +\sum_{j=0}^i \alpha^{[k]-[j]} \alpha_j x^{[j]}$.
\end{lemma}

\begin{proof}
By definition, 
\begin{eqnarray*}
P_{\alpha V} (x) &=& \prod_{u\in\alpha V}(x-u)\\
&=&\prod_{v\in V}(x-\alpha v)\\
&=&\alpha^{[k]}\prod_{v\in V}(\alpha^{-1}x-v)\\
&=&\alpha^{[k]}\cdot P_V(\alpha^{-1}x)\\
&=&x^{[k]} +\sum_{j=0}^i \alpha^{[k]-[j]} \alpha_j x^{[j]}.
\end{eqnarray*}
\end{proof}

\begin{lemma} \label{lemma:FrobeniusShiftPoly}
If $V\in \grsmn{q}{n}{k}$ and $P_V(x) = x^{[k]}+\sum _{j=0}^i \alpha _j x^{[j]}$ then for all $s\in \{0,\ldots,n-1\}$, $P_{F^s(V)}(x) = x^{[k]}+\sum _{j=0}^i F^s(\alpha _j) x^{[j]}$.
\end{lemma}

\begin{proof}
If $s \in \{0,\ldots,n-1\}$ and $u\in F^s(V)$ then $u=F^s(v)$ for some $v\in V$. Since $F^s$ is an automorphism, it follows that
\begin{eqnarray*}
u^{[k]}+\sum _{j=0}^i F^s(\alpha _j) u^{[j]} &=&  F^s(v)^{[k]}+\sum _{j=0}^i F^s(\alpha _j) F^s(v)^{[j]}\\
&=& F^s(v^{[k]})+\sum_{j=0}^iF^s(\alpha_j v^{[j]})=F^s\left( v^{[k]} + \sum _{j=0}^i \alpha_j v^{[j]}\right)\\
&=&F^s\left(P_V(v)\right)=F^s\left(\prod_{w\in V}(v-w)\right)= F^s(0) = 0.
\end{eqnarray*}
Therefore all elements of $F^s(V)$ are roots of $x^{[k]}+\sum _{j=0}^i F^s(\alpha _j) x^{[j]}$. Since the degree of this polynomial is $[k]$, the claim follows.
\end{proof}

The next lemma shows a connection between the coefficients of
the subspace polynomial of a given subspace $V\in\grsmn{q}{n}{k}$ and the
number of its distinct cyclic shifts. To formulate our claim,
we need the following equivalence relation.

\begin{definition}\label{definition:equivalenceRelation}
For $\alpha,\beta \in \bF_{q^n}^*$ and an integer $t$ which divides $n$, the equivalence relation $\sim_t$ is defined as follows
\[
\alpha \sim_t \beta \iff \frac{\alpha}{\beta} \in \bF_{q^t}.
\]
\end{definition}

Clearly, if $\alpha \sim_t \beta$ then $\alpha \in \beta \bF_{q^t}^* \cap \alpha \bF_{q^t}^*$, and since all the cyclic shifts of $\bF_{q^t}^*$ in $\bF_{q^n}^*$ are disjoint, it follows that $\beta \bF_{q^t}^*=\alpha \bF_{q^t}^*$. Hence, the equivalence classes under this relation are all the cyclic shifts of $\bF_{q^t}^*$ in $\bF_{q^n}^*$. Therefore, there are exactly $\frac{q^n-1}{q^t-1}$ equivalence classes of $\sim_t$ , each of which is of size $q^t-1$.

\begin{lemma} \label{lemma:distictCyclic}
Let $V\in\grsmn{q}{n}{k}$ and $P_V(x)=x^{[k]}+\sum_{j=0}^i \alpha_j x^{[j]}$. If $\alpha_s \ne 0$ for some $s\in \left\{1,\ldots,i\right\}$ and $\gcd(s,n)=t$ then $\alpha V\ne \beta V$ for all $\alpha,\beta \in \bF_{q^n}^*$ such that $\alpha\nsim _t \beta$.
\end{lemma}

\begin{proof}
Assume for contradiction that $\alpha V = \beta V$ for some $\alpha,\beta \in \bF_{q^n}^*$, where $\alpha \nsim _t \beta$. By Lemma~\ref{lemma:cyclicShiftPoly}
\begin{eqnarray*}
P_{\alpha V}(x) &=& x^{[k]}+\sum_{j=0}^i \alpha_j\cdot \alpha^{[k]-[j]} x^{[j]}\\
P_{\beta  V}(x) &=& x^{[k]}+\sum_{j=0}^i \alpha_j\cdot  \beta^{[k]-[j]} x^{[j]}.
\end{eqnarray*}
The equality $\alpha V = \beta V$, together with Lemma~\ref{lemma:subspaceEqual},  imply that
\begin{equation*}
\begin{cases}
        \begin{array}{lcl}
        \alpha_s \alpha^{[k]-[s]}&=&\alpha_s \beta^{[k]-[s]}\\
        \alpha _0 \alpha^{[k]-1}&=&\alpha_0 \beta^{[k]-1}
        \end{array}
\end{cases},
\end{equation*}
and since $\alpha_0\ne 0$ by Lemma~\ref{lemma:a0Not0}, it follows that 
\begin{eqnarray*}
\begin{cases}
        \begin{array}{lcl}
        \left( \frac{\alpha}{\beta}\right)^{[k]-[s]}&=&1\\
        \left( \frac{\alpha}{\beta}\right)^{[k]-1}&=&1
        \end{array}
\end{cases}.
\end{eqnarray*}
By dividing the second equation by the first equation, we get $\left(\frac{\alpha}{\beta}\right) ^{[s]-1}=1$. Hence, $\ord(\frac{\alpha}{\beta})\vert \gcd(q^n-1,q^s-1)$. It is well known that in $\bZ_{q^n-1}$, $\gcd(q^n-1,q^s-1)=q^{\gcd(n,s)}-1$ (e.g., \cite[p. 147, s. 38]{KGP94-Foundations}). Therefore, $\ord(\frac{\alpha}{\beta})\vert q^{\gcd (n,s)}-1$, which implies that $\frac{\alpha}{\beta}\in \bF_{q^t}$ since  $t=\gcd(n,s)$, and hence $\alpha \sim_t \beta$, a contradiction.
\end{proof}

\begin{corollary} \label{corollary:distinctCyclic}
Let $V\in\grsmn{q}{n}{k}$ and $P_V(x)=x^{[k]}+\sum_{j=0}^i \alpha_j x^{[j]}$. If $\alpha_s \ne 0$ for some $s\in \left\{1,\ldots,i\right\}$ with $\gcd(s,n)=t$ then~$V$ has at least $\frac{q^n-1}{q^t-1}$ distinct cyclic shifts.
\end{corollary}

To construct codes with more than one orbit using the Frobenius automorphism, one would like to find a sufficient condition
that a certain Frobenius shift is not a cyclic shift. Such
a condition can be derived for the special case,
where the subspace polynomial is a certain trinomial. The proof of the following lemma is deferred to \nameref{appendix:A}.

\begin{lemma} \label{lemma:CyclicIsNotForb}
If $V\in\grsmn{q}{n}{k}$ and $P_V(x) = x^{[k]}+\alpha_1 x^{[1]}+\alpha_0 x$,
where $\alpha_1 \ne 0$, then there exists
$\alpha \in \bF_{q^n}^*$, $i\in \{0,\ldots,n-1\}$
such that $F^i(V) = \alpha V$ if and only if
\[
\left(\frac{\alpha_0^{\frac{q^k-q}{q-1}}}{\alpha_1 ^{\frac{q^k-1}{q-1}}}\right)^{q^i-1}=1~.
\]
\end{lemma}

\section{Cyclic Subspace Codes} \label{section:constructions}
In this section some constructions of cyclic subspace codes
are provided. We distinguish between two cases. In Subsection~\ref{section:full_orbit}
we discuss codes whose codewords have a full length orbit.
In Subsection~\ref{section:deg_orbit} codes whose codewords
do not have a full length orbit are discussed.

\begin{definition}\label{definition:orbit}
Given a subspace $V\in\grsmn{q}{n}{k}$, the set $\{\alpha V\vert \alpha\in\bF_{q^n}^*\}$ is called \emph{the orbit of $V$}. The subspace $V$ has a \emph{full length orbit} if
$| \{\alpha V \vert \alpha \in \bF_{q^n}^* \}| = \frac{q^n-1}{q-1}$.
If $V$ does not have a full length orbit then it
has a \emph{degenerate orbit}.
\end{definition}

Note, that a cyclic code with a full length orbit cannot have a minimum distance $2k$. This is a simple observation from the fact that each element $\alpha\in \bF_{q^n}^*$ appears in exactly $\frac{q^k-1}{q-1}$ codewords.

We will give several simple related results on subspaces and the size of their orbits.
The first claim may be extracted from \cite[Corollary 3.13]{GMT14-StabSub}.  For completeness we include a shorter self-contained proof.

\begin{lemma}
\label{lemma:sizeOfCycles}
If $V\in\grsmn{q}{n}{k}$ then $| \{\alpha V~\vert~\alpha\in\bF_{q^n}^*\}|
=\frac{q^n-1}{q^t-1}$ for some $t$ which divides $n$.
\end{lemma}

\begin{proof} 
Let $\gamma$ be a primitive element in $\bF_{q^n}$ and let $\ell\in\bN$ be the smallest integer such that $\gamma^\ell V =V$. Clearly, $\ell\vert q^n-1$ and it is readily extracted that each $i\in\bN$ and each $0\le s <\ell$ satisfy $\gamma^s V = \gamma^{i\ell+s}V$. Furthermore, for every $s_1,s_2\in\{0,\ldots,\ell-1\}$ the sets $A_{s_j}\triangleq \{\gamma^{i\ell+s_j}~\vert~i\in\bN\}$ satisfy $|A_{s_1}|=|A_{s_2}|$. Let $\gamma^{i_1\cdot \ell},\gamma^{i_2\cdot \ell}\in A_0$ for some $i_1,i_2\in \bN$. Since $A_0 = \{\gamma^{i\ell}~\vert~i\in \bN\}$ it follows that
\[
\left(\gamma^{i_1\cdot \ell}+\gamma^{i_2\cdot \ell}\right)V \subseteq \gamma^{i_1\cdot \ell}V+\gamma^{i_2\cdot \ell}V=V+V=V,
\]
and hence $\gamma^{i_1\cdot \ell}+\gamma^{i_2\cdot \ell}\in A_0$, that is, $A_0$ is closed under addition. Since $A_0$ is also closed under multiplication, it follows that~$A_0$ is the multiplicative group of some subfield $\bF_{q^t}$ of $\bF_{q^n}$. Therefore, $|\{\alpha V~\vert~\alpha\in\bF_{q^n}^*\}|=\ell=\frac{q^n-1}{q^t-1}$.
\end{proof}

An immediate consequence of Lemma~\ref{lemma:sizeOfCycles} is that the largest possible size of an orbit is $\frac{q^n-1}{q-1}$, which justifies Definition~\ref{definition:orbit}. As will be shown in the sequel (see Section~\ref{section:deg_orbit}), the parameter $t$ from Lemma~\ref{lemma:sizeOfCycles} must also divide~$k$. A formula for the number of orbits of each possible size is given in~\cite{Drudge}. Most of the $k$-dimensional subspaces of $\bF_{q^n}$ have full length orbits. The
main goal in constructing cyclic codes is to obtain as many orbits as possible
in the code. This task will be left for future work. In this work we consider
first the existence of cyclic codes with one full length orbit and cyclic
codes with multiple full length orbits. Later, we consider the largest cyclic codes
for which all the orbits are degenerate.

%

\subsection{Codes with Full Length Orbits} \label{section:full_orbit}
\begin{lemma}\label{lemma:QisProduct} \cite[p. 107, Theorem~10]{MS-77} The polynomial $Q(x)\triangleq x^{[n]}-x$ is the product of all monic irreducible polynomials over $\bF_q$ with degree dividing $n$. 
\end{lemma}

\begin{theorem}
\label{theorem:irreducible}
If $q^k-1$ divides $n$ and $x^{[k]-1}+x^{[1]-1}+1$
is irreducible over $\bF_q$ then the polynomial $x^{[k]}+x^{[1]}+x$ is a subspace polynomial with respect to $\bF_{q^n}$.
\end{theorem}
\begin{proof}
Assume that $x^{[k]-1}+x^{[1]-1}+1$ is irreducible over $\bF_q$ and its degree divides $n$. By Lemma~\ref{lemma:QisProduct} $x^{[k]-1}+x^{[1]-1}+1\vert Q(x)$, and hence $x^{[k]}+x^{[1]}+x\vert Q(x)$. Therefore, $x^{q^k}+x^q+x$ is a subspace polynomial (see Definition \ref{definition:subspacePoly}), i.e., $P_V(x)=x^{[k]}+x^{[1]}+x$ for some subspace $V$. 
\end{proof}
\begin{corollary}
If $q^k-1$ divides $n$, $x^{[k]-1}+x^{[1]-1}+1$
is irreducible over $\bF_q$, and $V\in\grsmn{q}{n}{k}$ is the subspace whose subspace polynomial is $x^{[k]}+x^{[1]}+x$, then $\bC \triangleq \left\{ \alpha V~\vert~\alpha \in\bF_{q^n}^*\right\}$ is a cyclic subspace code of size $\frac{q^n-1}{q-1}$ and minimum distance at least $2k-2$.
\end{corollary}
\begin{proof}
According to Corollary \ref{corollary:distinctCyclic}, since the coefficient of $x^q$ in $P_V$ is nonzero, there are $\frac{q^n-1}{q-1}$ distinct cyclic shifts in $\bC$. By Lemma~\ref{lemma:cyclicShiftPoly} and Corollary \ref{corollary:distanceLowerBound}, the minimum distance of $\bC$ is at least $2k-2$.
\end{proof}

Although there exists an extensive research on irreducible trinomials over finite fields (e.g., \cite{vZG03-Trinomials}), no explicit construction of irreducible trinomials of the above form is known. However, the following examples were easily found using a computer search.

\begin{example}
Since the polynomials $x^{2^k-1}+x+1$ are irreducible over $\bF_2$ for all $k\in\{2,3,4,6,7,15\}$, it follows that the polynomial $x^{2^k}+x^2+x$ is a subspace polynomial of a subspace $V\in\grsmn{2}{(2^k-1)t}{k}$ for all $t\in\bN$. Therefore, the code $\bC\triangleq\{\alpha V~\vert~\alpha\in\bF_{q^{(2^k-1)t}}^*\}$ is cyclic code of size $2^{t\cdot (2^k-1)}-1$ and minimum distance $2k-2$ in~$\grsmn{2}{(2^k-1)t}{k}$.
\end{example}

By using a similar approach we have that for any $k$ and $q$,
cyclic codes in $\grsmn{q}{n}{k}$ can be explicitly constructed for
infinitely many values of $n$. The construction will make use of the following lemma.

\begin{lemma}\label{lemma:splittingField}
If $f(x)=\prod_{i=1}^t p_i^{\alpha_i}(x)$ is a polynomial over $\bF_q$ and $p_1(x),\ldots,p_t(x)$ are its irreducible factors in $\bF_q$ then $f(x)$ splits completely in $\bF_{q^n}$ for $n = \lcm\{\deg p_i(x)\}_{i=1}^t$.\footnote{$\lcm\{s_i\}_{i=1}^{t}$ stands for the least common multiplier of the integers $s_1,\ldots,s_t$.}
\end{lemma}

\begin{proof}
According to \cite[Corollary 2.15, p. 52]{LN97-FiniteFields}, the splitting field of an irreducible polynomial of degree $m$ over $\bF_q$ is $\bF_{q^m}$. Therefore, for each $i=1,\ldots,t$, the splitting field of $p_i$ is $\bF_{q^{n_i}}$, where $n_i\triangleq\deg p_i$. For any $i$, the only finite fields that contain $\bF_{q^{n_i}}$ are of the form $\bF_{q^r}$ for $r$ such that $n_i|r$. Hence, the smallest field that contains $\bF_{q^{n_i}}$ for all $i$ is $\bF_{q^n}$.
\end{proof}

\begin{theorem} \label{theorem:explicit2kminus2}
For any $k$ and $q$ we may explicitly construct a cyclic subspace code of size $\frac{q^n-1}{q-1}$ and minimum distance $2k-2$ in $\grsmn{q}{n}{k}$ for infinitely many values of $n$.
\end{theorem}

\begin{proof}
By factoring $T(x)\triangleq x^{[k]}+x^{[1]}+x$ and computing the least common multiplier of the degrees of its factors we find the degree of the splitting field of $T(x)$ (see Lemma~\ref{lemma:splittingField}). The subspace $V$, whose corresponding subspace polynomial is $T(x)$ may be easily found by finding the kernel of the linear transformation defined by $T$. If $\bC \triangleq \left\{ \alpha V~\vert~\alpha \in\bF_{q^n}^*\right\}$ then by Corollary \ref{corollary:distinctCyclic} there are $\frac{q^n-1}{q-1}$ distinct cyclic shifts in $\bC$. By Lemma~\ref{lemma:cyclicShiftPoly} and Corollary \ref{corollary:distanceLowerBound}, the minimum distance of $\bC$ is at least $2k-2$. Infinitely many values of $n$ will are by considering the cyclic shifts of $V$ in all the field extensions of the splitting field.
\end{proof}

\begin{remark}
Theorem~\ref{theorem:explicit2kminus2} proves Conjecture \ref{conjecture:2kminus2} for infinitely many values of $n$.
\end{remark}

\begin{remark}
The codes implied by Theorems~\ref{theorem:irreducible} and Theorem~\ref{theorem:explicit2kminus2} cannot be enlarged using the Frobenius isomorphism due to Lemma~\ref{lemma:CyclicIsNotForb}, since for any $i\in\{0,\ldots,n-1\}$ we have that the $i$th Frobenius shift is also a cyclic shift.
\end{remark}

Let $N=t\cdot n$ and let $\gamma$ be a primitive element
in $\bF_{q^N}$. Note, that the set
$\{0\}\cup \{\gamma^{i(q^N-1)/(q^n-1)}\} _{i=0}^{q^n-2}$
is the unique subfield $\bF_{q^n}$ of $\bF_{q^N}$.
Let $V$ be a subspace of $\bF_{q^n}$. Since $\bF_{q^n}\subseteq \bF_{q^N}$
we can view the subspace $V$ as a subspace of $\bF_{q^N}$ over $\bF_q$. 

Now, we present a general method for constructing cyclic codes in $\grsmn{q}{N}{k}$, where $N=t\cdot n$ for some prime $n$, which have more than one full length orbit. We do so by using the Frobenius automorphism.

\begin{lemma} \label{lemma:multipleOrbitUsingPrimality}
Let $n$ be a prime, $n\vert N$, $V\in\grsmn{q}{N}{k}$ and $P_V(x)=x^{[k]} +\alpha_1 x^{[1]}+\alpha_0 x$, where $\alpha_0,\alpha_1 \in \bF_{q^n}^*$. If $\alpha_1 ^{\frac{q^k-1}{q-1}} \nsim_1 \alpha_0^{\frac{q^k-q}{q-1}}$ (see Definition \ref{definition:equivalenceRelation}) then the code $\bC\sus \grsmn{q}{N}{k}$ defined by
\begin{eqnarray}\label{equation:multipleOrbitsDefinition}
\bC \triangleq \bigcup_{i=0}^{n-1} \left\{\alpha \cdot F^i(V)~\vert~ \alpha \in \bF_{q^N}^*\right\}
\end{eqnarray}
is of size $n\cdot\frac{q^N-1}{q-1}$ and minimum distance $2k-2$.
\end{lemma}

\begin{proof} 
The code $\bC$ is obviously cyclic. By Lemmas~\ref{lemma:intersectionLemma},~\ref{lemma:cyclicShiftPoly}, and \ref{lemma:FrobeniusShiftPoly}, the dimension of the intersection between any two distinct subspaces in $\bC$ is at most 1, and hence the minimum distance of $\bC$ is $2k-2$. 

To show that $|\bC|=n\cdot\frac{q^N-1}{q-1}$, fix $i$ and notice that by Lemma~\ref{lemma:FrobeniusShiftPoly} we have that the coefficient of $x^{[1]}$ in $P_{F^i(V)}(x)$ is non-zero. Therefore, Lemma~\ref{lemma:distictCyclic} implies that the set $\{\alpha \cdot F^i(V)~\vert~ \alpha \in \bF_{q^N}^*\}$ consists of $\frac{q^N-1}{q-1}$ distinct subspaces. 

To complete the proof, we have to show that all the sets in the union in (\ref{equation:multipleOrbitsDefinition}) are disjoint. Let $i, j \in \{0,\ldots,n-1\},i\ne j$, and assume for contradiction that there exists $\beta,\gamma \in \bF_{q^N}^*$ such that $\beta F^i (V) = \gamma F^j(V)$. W.l.o.g assume that $j>i$, and denote $U\triangleq F^i(V)$. Notice that by Lemma~\ref{lemma:FrobeniusShiftPoly} we have
\[
P_U(x) = P_{F^i(V)}(x)= x^{[k]}+F^i(\alpha_1) \cdot x^{[1]}+F^i(\alpha_0)\cdot x= x^{[k]}+\alpha_1 ^{[i]} \cdot x^{[1]}+\alpha_0 ^{[i]}\cdot x.
\]
Since $F^{j-i}(U)=\frac{\beta}{\gamma}\cdot U$, we may apply Lemma~\ref{lemma:CyclicIsNotForb} to get
\begin{eqnarray} \label{equation:CondOnU}
\left( \frac{\left(\alpha_0^{q^{i}}\right)^{\frac{q^k-q}{q-1}}}{\left(\alpha_1^{q^{i}}\right)^{\frac{q^k-1}{q-1}}}\right)^{q^{j-i}-1} = 1.
\end{eqnarray}
Denote $z\triangleq \frac{\alpha_0^{\frac{q^k-q}{q-1}}}{\alpha_1^{\frac{q^k-1}{q-1}}}$ and notice that
\begin{enumerate}
\item[A1.] Equation (\ref{equation:CondOnU}) implies $z^{q^{i}(q^{j-i}-1)}=1$.
\item[A2.] The condition $\alpha_1 ^{\frac{q^k-1}{q-1}}\nsim_1\alpha_0^{\frac{q^k-q}{q-1}}$ implies $z \notin \bF_q$.
\item[A3.] Since $\alpha_0,\alpha_1 \in \bF_{q^n}^*$ it follows that $z \in \bF_{q^n}^*$.
\end{enumerate}
By A1 and A3 we have that $\ord(z)$ divides both $q^{i}(q^{j-i}-1)$ and $q^n-1$, therefore $\ord(z)\vert gcd(q^{i}(q^{j-i}-1),q^n-1)$. Since $q^n-1$ is not a power of $q$, it follows that $gcd(q^n-1,q^{i})=1$, and hence,
\[
gcd(q^{i}(q^{j-i}-1),q^n-1) = gcd(q^{j-i}-1,q^n-1).
\]
It is well known that in any field $gcd(x^r-1,x^s-1)=x^{gcd(r,s)}-1$ (e.g., \cite[p. 147, s. 38]{KGP94-Foundations}). Therefore, the primality of~$n$ implies that $\gcd(q^{j-i}-1,q^n-1)=q^{\gcd(j-i,n)}-1=q-1$, and hence $\ord(z)\vert q-1$. The only elements of $\bF_{q^N}$ whose order divides $q-1$ are the elements of $\bF_q$, and hence $z\in\bF_q$, a contradiction to A2.
\end{proof}

Lemma~\ref{lemma:alpha0alpha1} which follows, whose proof is deferred to \nameref{appendix:A}, shows that coefficients $\alpha_0,\alpha_1$ from Lemma~\ref{lemma:multipleOrbitUsingPrimality} may be easily found in $\bF_{q^n}$.

\begin{lemma}\label{lemma:alpha0alpha1}
Let $n$ be prime and let $\gamma$ be a primitive element in $\bF_{q^n}$. If $\alpha_0\triangleq\gamma$ and $\alpha_1 \triangleq \gamma^q$ then $\alpha_1 ^{\frac{q^k-1}{q-1}} \nsim_1 \alpha_0^{\frac{q^k-q}{q-1}}$.
\end{lemma}

As a consequence of Lemma~\ref{lemma:multipleOrbitUsingPrimality} and Lemma~\ref{lemma:alpha0alpha1} we have the following theorem.

\begin{theorem}\label{theorem:multipleOrbits}
Let $n$ be prime, $\gamma$ a primitive element of $\bF_{q^n}$, and define $\alpha_0\triangleq \gamma$ and $\alpha_1 \triangleq \gamma^q$. If $\bF_{q^N}$ is the splitting field of the polynomial $x^{[k]} +\alpha_1 x^{[1]}+\alpha_0 x$ and $V\in\grsmn{q}{N}{k}$ its corresponding subspace, then  
\begin{eqnarray*}
\bC \triangleq \bigcup_{i=0}^{n-1} \left\{\alpha \cdot F^i(V)~\vert~ \alpha \in \bF_{q^N}^*\right\}
\end{eqnarray*}
is a cyclic code of size $n\cdot \frac{q^N-1}{q-1}$ and minimum distance $2k-2$.
\end{theorem}
Note that the construction in Theorem~\ref{theorem:multipleOrbits} improves the construction of Theorem~\ref{theorem:explicit2kminus2}. In Theorem~\ref{theorem:explicit2kminus2} we construct a code with one full length orbit, where in Theorem~\ref{theorem:multipleOrbits} we add multiple orbits without compromising the minimum distance.

\subsection{Codes with degenerate orbits}\label{section:deg_orbit} 
In this subsection it is shown that subspaces of $\grsmn{q}{n}{k}$ that may be considered as subspaces over a subfield of $\bF_{q^n}$ which is larger than $\bF_q$, form a cyclic code with a unique subspace polynomial structure. The cyclic property and the minimum distance of this code are an immediate consequence of this unique structure.

\begin{lemma}\label{lemma:embedding}
If $n,k\in \bN,k<n$ and $d\in \bN$ divides $\gcd(n,k)$, then there exists an $\bF_{q^d}$-homomorphism from $\grsmn{q^d}{n/d}{k/d}$ to $\grsmn{q}{n}{k}$.
\end{lemma}

\begin{proof}
Let $\bF_{q^d}^{n/d}$ be the vector space of dimension $n/d$ over $q^d$. It is widely known that there exists an isomorphism $f$ from $\bF_{q^d}^{n/d}$ to $\bF_{(q^d)^{n/d}}$. Notice that by our abuse of notation, both $\bF_{q^d}^{n/d}$ and $\bF_{(q^d)^{n/d}}$ can be considered as vector spaces over $\bF_{q^d}$. Since there is a unique field with $q^n$ elements, $\bF_{q^n}$ may also be considered as a vector space over $\bF_{q^d}$. Therefore, there exists an isomorphism $g:\bF_{q^d}^{n/d}\to\bF_{q^n}$ such that $g\triangleq h\circ f$, where $h$ is some isomorphism from $\bF_{(q^d)^{n/d}}$ to $\bF_{q^n}$.

Notice that for all $u,v\in\bF_{q^d}^{n/d}$ and $\alpha,\beta\in\bF_{q^d}$, we have $g(\alpha v+\beta u)=\alpha g(v)+\beta g(u)$. For $V\in \grsmn{q^d}{n/d}{k/d}$ let $G(V)\triangleq\{g(v)|v\in V\}$. The set $G(V)$ is clearly a subspace of dimension $k$ over $\bF_q$ in $\bF_{q^n}$. Furthermore, the function $G:\grsmn{q^d}{n/d}{k/d}\to\grsmn{q}{n}{k}$ is injective since $g$ is injective.
\end{proof}

Lemma \ref{lemma:embedding} allows us to define the following set of subspaces.

\begin{construction}\label{construction:Cd}
For $n,k\in \bN$ and $d\in \bN$ such that $d|\gcd(n,k)$, let $\bC_d$ be the code \[\{G(V)\vert V\in\grsmn{q^d}{n/d}{k/d}\},\]where $G$ was defined in the proof of Lemma~\ref{lemma:embedding}.
\end{construction}
Since $\bC_d$ is the image of an injective function from $\grsmn{q^d}{n/d}{k/d}$ to $\grsmn{q}{n}{k}$, we have the following. 
\begin{corollary}\label{corollary:CdSize}
$|\bC_d|=\qbin{n/d}{k/d}{q^d}$.
\end{corollary}

\begin{remark}\label{remark:sumOfFqd}
The code $\bC_d$ from construction \ref{construction:Cd} may be alternatively defined as
\[\bC_d \triangleq \left\{\sum_{i=1}^{k/d}\alpha_i \bF_{q^d}~\Big\vert~ \alpha_1,\ldots,\alpha_{k/d} \in \bF_{q^n} \text{ are linearly independent over }\bF_{q^d}\right\}.\]
The proof of the equivalence of this alternative definition appears in \nameref{appendix:B}. The code $\bC_d$ may also be defined as the set of all subspaces of $\grsmn{q}{n}{k}$ that are also subspaces over~$\bF_{q^d}$.
\end{remark} 
The subspaces in $\bC_d$ admit a unique subspace polynomial structure, from which the useful properties of $\bC_d$ are apparent. 

\begin{lemma}\label{lemma:CdPolyStructure}
If $V\in\grsmn{q}{n}{k}$ then $V\in \bC_d$ if and only if $P_V(x)=\sum_{i=0}^{k/d}c_i x^{[di]}$ for some $c_i$'s in $\bF_{q^n}$. 
\end{lemma}
\begin{proof}
Let $V\in\bC_d$, and let $U\in\grsmn{q^d}{n/d}{k/d}$ be such that $F(U)=V$ (see Construction~\ref{construction:Cd}). By Definition \ref{definition:subspacePoly} it follows that $P_U|x^{(q^d)^{n/d}}-x$. Since $x^{(q^d)^{n/d}}-x=x^{[n]}-x$, it follows that $P_U$ is a subspace polynomial of a subspace $W\in\grsmn{q}{n}{k}$. The roots of $P_U$ are precisely the set $\{g(u)\vert u\in U\}$, where $g$ is the isomorphism between $\bF_{q^d}^{n/d}$ and $\bF_{q^n}$ mentioned in the proof of Lemma~\ref{lemma:embedding}, and hence, $W=V$. Since $P_U$ is a subspace polynomial of a subspace in $\grsmn{q^d}{n/d}{k/d}$, its subspace polynomial is of the form $P_U(x)=\sum_{i=0}^{k/d}c_i x^{(q^d)^i}$. Since $P_V=P_U$, the claim follows.

Conversely, let $V\in \grsmn{q}{n}{k}$ with $P_V(x)=\sum_{i=0}^{k/d}c_i x^{[di]}$. By Definition \ref{definition:subspacePoly}, it follows that $P_V|x^{[n]}-x$, and thus ${P_V|x^{(q^d)^{n/d}}-x}$. Therefore $P_V$ is a subspace polynomial of some $U\in\grsmn{q^d}{n/d}{k/d}$, and hence $V\in\bC_d$.
\end{proof}

\begin{corollary} \label{corollary:CdCyclic}
$\bC_d \subseteq \grsmn{q}{n}{k}$ is a cyclic subspace code.
\end{corollary}

\begin{proof}
Let $V\in\bC_d$ and $\alpha\in\bF_{q^n}^*$. By Lemma \ref{lemma:CdPolyStructure} the subspace polynomial of $V$ is of the form $P_V(x)=\sum_{i=0}^{k/d}c_i x^{[di]}$ for some $c_i\in \bF_{q^n}$. By Lemma \ref{lemma:cyclicShiftPoly} the subspace polynomial of $\alpha V$ is $P_V(x)=\sum_{i=0}^{k/d}c_i\alpha^{[k]-[di]} x^{[di]}$. Again by Lemma~\ref{lemma:CdPolyStructure}, it follows that $\alpha V\in \bC_d$.
\end{proof}

Since for $V\in\bC_d$ we have that $\gap(V)\ge d$, and the following result is a consequence of Corollary~\ref{corollary:distanceLowerBound} and Definition~\ref{definition:gap}.

\begin{corollary}\label{corollary:CdGap}
The minimum distance of $\bC_d$ is $2d$.
\end{corollary}

The structure of the subspace polynomials of the codewords of $\bC_d$ allows us to construct a code~$\bC$ which is a union of $\bC_{d_i}$ for distinct $d_i$'s which divide $\gcd(n,k)$. We now analyze the size and distance of the resulting code.

\begin{lemma}\label{lemma:CdiIntersection}
Let $k,n\in \bN,k<n$. If $d_1,\ldots,d_t$ divide both $n$ and $k$ and $d=\lcm (d_1,\ldots,d_t)$ then~$\bigcap _{i=1}^t \bC_{d_i} = \bC_d$. 
\end{lemma}
\begin{proof}
According to Lemma~\ref{lemma:CdPolyStructure} if $V\in \bC_d$ then $P_V(x)=\sum_{i=0}^{k/d} c_i x^{[id]}$. Since $d_j \vert d$ for each $j$, we may also write $P_V(x)=\sum _{i=0}^{k/d_j} c_i' x^{[id_j]}$, where all additional coefficients are 0, and thus $V\in \bC_{d_j}$ for each~$j$. 

On the other hand, if $V\in \bigcap_{i=1}^t \bC_{d_i}$, again by Lemma~\ref{lemma:CdPolyStructure} it follows that all nonzero coefficients of $P_V$ correspond to $x^{[\ell]}$ such that $d_j\vert \ell$ for each $j$. Thus, $d\vert \ell$ and $V\in \bC_d$.
\end{proof}

\begin{construction}\label{construction:manyCds}
Let $k,n\in \bN,k<n$. If $d_1,\ldots,d_t$ divide both $n$ and $k$ then let $\bC \triangleq \bigcup_{i=1}^t \bC_{d_i}$.
\end{construction}

\begin{lemma}
$\bC$ is a cyclic code of with codewords of dimension $k$ and minimum distance $2\min\{d_i\}_{i=1}^t$. The size of $\bC$ is given by 
\[
|\bC| = \sum_{i=1}^t\qbin{n/d_i}{k/d_i}{q^{d_i}}-\sum_{i<j}\qbin{n/\lcm(d_i,d_j)}{k/\lcm(d_i,d_j)}{q^{\lcm(d_i,d_j)}}+\sum_{i<j<\ell}\qbin{n/\lcm(d_i,d_j,d_\ell)}{k/\lcm(d_i,d_j,d_\ell)}{q^{\lcm(d_i,d_j,d_\ell)}}-\cdots.
\]
\end{lemma}

\begin{proof}
By Corollary \ref{corollary:CdGap} we have that $gap(V)\ge \min\{d_i\}_{i=1}^t$ for each $V\in\bC$, and hence the minimum distance of $\bC$ is at least $2\min\{d_i\}_{i=1}^t$ by Corollary \ref{corollary:distanceLowerBound}. By Corollary~\ref{corollary:CdSize} we have that $|\bC_{d_i}|=\qbin{n/d_i}{k/d_i}{q^{d_i}}$ for each $i$. Furthermore, by Lemma~\ref{lemma:CdiIntersection} the size of the intersection of $\bC_{d_{i_1}},\ldots,\bC_{d_{i_\ell}}$ is $\qbin{n/d}{k/d}{q^d}$ where~$d=\lcm(d_{i_1},\ldots,d_{i_\ell})$. These facts allow us to obtain the exact size of $\bC$ using the inclusion-exclusion principle \cite[Chapter~10]{vanLintandWilson}.
\end{proof}

Using similar techniques, we show that a cyclic code over a large field may be embedded in a Grassmannian over a smaller field, while preserving cyclicity and multiplying the minimal distance by some factor. Note that Construction \ref{construction:Cd} is a special case of this technique, where the embedded code is $\grsmn{q^d}{n/d}{k/d}$. 

\begin{theorem}
Let $d$ be an integer such that $d|\gcd(n,k)$. If $\bC\subseteq \grsmn{q^d}{n/d}{k/d}$ is a cyclic code with minimum distance $2\cdot (k/d) -2\delta$ then there exists a cyclic code $\bC' \subseteq \grsmn{q}{n}{k}$ of size $|\bC|$ and minimum distance $2k-2d \delta$.
\end{theorem}

\begin{proof}
Let $g:\bF_{q^d}^{n/d}\to \bF_{q^n}$ and $G:\grsmn{q^d}{n/d}{k/d}\to\grsmn{q}{n}{k}$ be the embeddings defined in the proof of Lemma~\ref{lemma:embedding}. If $\bC'\triangleq\{G(V)\vert V\in\bC\}$ then $|\bC'|=|\bC|$, since $G$ is injective. The cyclic property of $\bC'$ follows from the fact that $P_V(x)=P_{G(V)}(x)$ for all $V\in\bC$, as shown in the proof of Lemma ~\ref{lemma:CdPolyStructure}. To bound the minimum distance of $\bC'$ it suffices to show that if $U_1,U_2\in \bC$ then \[\dim \left(G(U_1) \cap G(U_2)\right)= d \cdot\dim(U_1,U_2).\]

Indeed, if $w\triangleq\dim(U_1\cap U_2)$, then since $g$ is an isomorphism of subspaces over $\bF_{q^d}$, it follows that the set $Z\triangleq\{g(z)\vert z\in U_1\cap U_2\}$ is a subspace of $\bF_{q^n}$ over $\bF_q$. By a simple counting argument, ${\dim Z = dw}$, and hence, $\dim \left(F(U_1) \cap F(U_2)\right)\ge dw$. Assuming for contradiction that $\dim \left(F(U_1) \cap F(U_2)\right)>dw$ clearly implies that $\dim(U_1\cap U_2)>w$, a contradiction.
\end{proof}

\section{Conclusions and Future Work}\label{section:Conclusions}

In this paper we have considered constructions of cyclic
subspace codes. We have proved the existence of a cyclic code
in $\grsmn{q}{n}{k}$ for any given $k$ and infinitely many values of~$n$.
The constructed codes have minimum subspace distance $2k-2$,
the normalizer of a Singer subgroup is their automorphism group
if~$n$ is a prime, and they have full length orbits for all values of~$n$.
We have also constructed large codes when all the orbits are
degenerated. We have shown how the representation of subspaces by their
subspace polynomials can be used in constructing subspace codes. 

For future research, the main problems are to construct cyclic codes of large size, to explore the structure and properties of our codes, and to examine possible decoding algorithms for them. It is easily verified that the vast majority of subspaces have full length orbits. Therefore, it seems reasonable to conjecture that full length orbits with minimum distance $2k-2$ exist for any value of $n,k,q$ (see Conjecture \ref{conjecture:2kminus2}). Although the codes presented in Section~\ref{section:full_orbit} are the first known explicit construction of such codes, they are most likely the tip of the iceberg, and codes of these parameters are abound. 

Although the gap of two polynomials provides significant information about the intersection of their respective subspaces, Remark \ref{remark:gapNotTight} shows that the gap might not be the most efficient tool for this purpose. Therefore, another open problem is finding a better measure for the intersection of two subspaces, and in particular, two subspaces from the same orbit.

A prominent part of the study of subspace polynomials relies on understanding the connection between the coefficients of a polynomial and the size of the respective splitting field. Hence, any progress in this direction may provide an improvement of our results. 

\section*{Acknowledgments}
The authors would like to thank Thomas Honold for bringing~\cite{Drudge} to their attention.

\section*{Appendix A}\label{appendix:A}
\begin{proof} (of Lemma~\ref{lemma:CyclicIsNotForb})
Assume $F^i(V)=\alpha V$ for some $\alpha$. By Lemmas~\ref{lemma:cyclicShiftPoly} and \ref{lemma:FrobeniusShiftPoly},
\begin{eqnarray*}
P_{\alpha V} (x) &=& x^{[k]}+\alpha^{[k]-[1]}\cdot\alpha_1 x^{[1]}+\alpha^{[k]-1}\cdot\alpha_0 x\\
P_{F^i(V)} (x) &=& x^{[k]}+F^i(\alpha_1) x^{[1]}+F^i(\alpha_0) x.
\end{eqnarray*}
By Lemma~\ref{lemma:subspaceEqual},

\begin{equation*}
\begin{cases}
        \begin{array}{lcl}
        \alpha^{[k]-[1]}\cdot\alpha_1=F^i(\alpha_1)\\
        \alpha^{[k]-1}\cdot\alpha_0=F^i(\alpha_0)
        \end{array}
 \\
\end{cases}
\end{equation*}
\begin{equation*}
\begin{cases}
        \begin{array}{lcl}
        \alpha^{[k]-[1]}\cdot\alpha_1=\alpha_1^{[i]}\\
        \alpha^{[k]-1}\cdot \alpha_0=\alpha_0^{[i]}
        \end{array}
 \\
\end{cases}
\end{equation*}
Since $\alpha_0\ne 0$ (by Lemma~\ref{lemma:a0Not0}) and $\alpha_1\ne 0$, $\alpha ^{[1]-1}=\left(\frac{\alpha_0}{\alpha_1}\right) ^{q^i-1}$. Using some algebraic manipulations we have,
\begin{eqnarray*}
\alpha^{[k]-[1]}&=&\alpha_1^{[i]-1}\\
\alpha^{\left(q-1\right)\left(\frac{q^k-q}{q-1}\right)}&=&\alpha_1^{q^i-1}\\
\left(\frac{\alpha_0}{\alpha_1}\right) ^{\left(q^i-1\right)\left(\frac{q^k-q}{q-1}\right)}&=&\alpha_1^{q^i-1}\\
\frac{\alpha_0^{\frac{q^k-q}{q-1}\cdot(q^i-1)}}{\alpha_1^{\frac{q^k-q}{q-1}\cdot(q^i-1)+(q^i-1)}}&=&1\\
\frac{\alpha_0^{\frac{q^k-q}{q-1}\cdot(q^i-1)}}{\alpha_1^{\frac{q^k-1}{q-1}\cdot(q^i-1)}}&=&1\\
\left(\frac{\alpha_0^{\frac{q^k-q}{q-1}}}{\alpha_1^{\frac{q^k-1}{q-1}}}\right)^{q^i-1}&=&1,\\
\end{eqnarray*}

which concludes the proof of one direction of the lemma. Now assume \[\left(\frac{\alpha_0^{\frac{q^k-q}{q-1}}}{\alpha_1 ^{\frac{q^k-1}{q-1}}}\right)^{q^i-1}=1.\]Define $\alpha \triangleq \left( \frac{\alpha_0}{\alpha_1}\right)^{\frac{q^i-1}{q-1}}$. We get 
\begin{eqnarray*}
\left(\frac{\alpha_0^{\frac{q^k-q}{q-1}}}{\alpha_1 ^{\frac{q^k-1}{q-1}}}\right)^{q^i-1}&=&1\\
\left(\frac{\alpha_0^{\frac{q^k-q}{q-1}}}{\alpha_1 ^{\frac{q^k-q}{q-1}}\cdot \alpha_1}\right)^{q^i-1}&=&1\\
\left(\frac{\alpha_0^{\frac{q^i-1}{q-1}}}{\alpha_1 ^{\frac{q^i-1}{q-1}}}\right)^{q^k-q}&=&\alpha_1 ^{q^i-1}\\
\alpha ^{q^k-q}&=&\alpha_1 ^{q^i-1}.
\end{eqnarray*}
In addition, we have $\alpha^{q^k-1}=\alpha^{q^k-q}\alpha^{q-1}=\left(\alpha_1^{q^i-1} \right)\cdot\left(\frac{\alpha_0^{q^i-1}}{\alpha_1 ^{q^i-1}}\right)=\alpha_0^{q^i-1}$. Therefore:
\begin{equation*}
\begin{cases}
        \begin{array}{lcl}
        \alpha^{q^k-q}=\alpha_1^{q^i-1}\\
        \alpha^{q^k-1}=\alpha_0^{q^i-1}
        \end{array}
 \\
\end{cases}
\end{equation*}
\begin{equation*}
\begin{cases}
        \begin{array}{lcl}
        \alpha^{q^k-q}\cdot \alpha_1=\alpha_1^{q^i}\\
        \alpha^{q^k-1}\cdot \alpha_0=\alpha_0^{q^i}
        \end{array}
 \\
\end{cases},
\end{equation*}
which implies that $F^i(V)=\alpha V$ due to equality between the coefficients of the corresponding subspace polynomials.
\end{proof}

\begin{proof} (of Lemma~\ref{lemma:alpha0alpha1}) Assume for contradiction that 
\[
\alpha_0^{\frac{q^k-q}{q-1}}\sim_1 \alpha_1 ^{\frac{q^k-1}{q-1}},
\]
i.e., there exists $\alpha\in \bF_q^*$ such that 
\begin{eqnarray}\label{equation:alpha1}
\alpha\cdot\gamma^{\frac{q^k-q}{q-1}}=\left(\gamma^q\right)^{\frac{q^k-1}{q-1}}.
\end{eqnarray}
Raising both sides of (\ref{equation:alpha1}) by the $(q-1)$th power yields
\begin{eqnarray}
\nonumber\gamma^{q^k-q}&=&\gamma^{q^{k+1}-q}\\
\gamma^{q^k(q-1)}&=&1.\label{equation:alpha2}
\end{eqnarray}
Since $q\in\bZ_{q^n-1}^*$, it follows that $q$ has a multiplicative inverse $w$ modulo $q^n-1$. By raising both sides of (\ref{equation:alpha2}) by the $w^k$th power we get that $\gamma^{q-1}=1$, and hence, $\gamma\in \bF_q$, a contradiction.
\end{proof}

\section*{Appendix B}\label{appendix:B}
In this appendix we prove the equivalence of an alternative definition to Construction \ref{construction:Cd} (see Remark~\ref{remark:sumOfFqd}). The following lemma is required for the proof of equivalence.

\begin{lemma}\label{lemma:unionIsDirectSum}
If $V\in\grsmn{q}{n}{k}$ may be written as $V=\bigcup_{i=1}^{\ell}\alpha_i \bF_{q^d}$, where $d|\gcd(n,k)$, then $V$ may be written as a \textit{direct sum} $V=\sum_{j=1}^{k/d}\beta_j \bF_{q^d}$.
\end{lemma}

\begin{proof}
We show that for every $J\subseteq \{1,\ldots,\ell\}$, every $\alpha_i \bF_{q^d}$ is either contained or mutually disjoint with $A_J\triangleq \sum_{j\in J}\alpha_j \bF_{q^d}$. Assume for contradiction that there exists $\alpha_i\bF_{q^d},i\notin J$ that is neither contained nor mutually disjoint with $A_J$. That is, there exists $u_1,u_2 \in \bF_{q^d}^*$ such that $\alpha_i u_1\in A_J$ and $\alpha_i u_2\notin A_J$. Since $\alpha_i u_1\in A_J$ it follows that there exists $s_j\in\bF_{q^d}$ for each $j\in J$ such that $\alpha_i u_1 =\sum_{j\in J}\alpha_j s_j$. Hence, $\alpha_i = u_1^{-1}\sum_{j\in J}\alpha_j s_j$ and therefore $\alpha_i u_2 = u_1^{-1}\sum_{j\in J}\alpha_j s_j \cdot u_2$. However, since $u_2/u_1 \in \bF_{q^d}^*$ and $s_j \in \bF_{q^d}$ for all $j$, it follows that $\alpha_i u_2 = \sum_{j\in J}\alpha_j (s_j u_2/u_1)\in A_J$, a contradiction. 
Therefore, by taking $\alpha_1\bF_{q^d}$ and iteratively expanding it by adding disjoint cyclic shifts of $\bF_{q^d}$, the required direct sum may be achieved.
\end{proof}

\begin{theorem}\label{theorem:CdAlternative}
Let $d\in \bN$ such that $d|\gcd(n,k)$. For a subspace $V\in\grsmn{q}{n}{k}$, $V\in \bC_d$ (see Construction \ref{construction:Cd}) if and only if $V$ may be written as a direct sum of cyclic shifts of $\bF_{q^d}$.
\end{theorem}

\begin{proof}
If $V\in\bC_d$ then by Lemma~\ref{lemma:CdPolyStructure}, $P_V(x)=\sum_{i=0}^{k/d}c_ix^{[id]}$. Since for all $\gamma\in\bF_{q^d}$ we have $\gamma^{[d]-1}=1$, it follows that if $P_V(v)=0$ for $v\in\bF_{q^n}$, then $P_V(\gamma v)=\sum_{i=0}^{k/d}c_i(\gamma v)^{[id]}=P_V(v)=0$. Therefore, $V$ is a union of cyclic shifts of $\bF_{q^d}$, and according to Lemma~\ref{lemma:unionIsDirectSum} may be written as a direct sum of cyclic shifts of $\bF_{q^d}$.

On the other hand, if $V=\sum_{i=1}^{k/d}\alpha_i\bF_{q^d}$ such that $\alpha_i\in\bF_{q^n}$, let $\beta_i\in\bF_{q^d}^{n/d}, i\in\{1,\ldots,k/d\}$ such that $\beta_i\triangleq g^{-1}(\alpha_i)$, where $g$ is the isomorphism between $\bF_{q^d}^{n/d}$ and $\bF_{q^n}$ mentioned in the proof of Lemma~\ref{lemma:embedding}. Let $U$ be the linear span of~$\{\beta_i\}_{i=1}^{k/d}$ in $\bF_{q^d}^{n/d}$ over $\bF_{q^d}$. We show that $U$ is a $\frac{k}{d}$-subspace. Assume for contradiction that the elements of $\{\beta_i\}_{i=1}^{k/d}$ are linearly dependent, i.e., there exists $\gamma_i\in\bF_{q^d}$ such that $\sum \gamma_i\beta_i=0$. Hence, $0=g(0)=g(\sum \gamma_i \beta_i)=\sum \gamma_i \alpha_i$ and therefore, the element $0\in\bF_{q^n}$ has two distinct representations as an element of $V$. This implies that $|V|<q^k$, a contradiction. Now observe that,
\begin{eqnarray*}
G(U)&=&\left\{g(u)~\vert~u\in U\right\}\\
&=&\left\{g\left(\sum \gamma_i \beta_i\right) ~\vert~\forall i, \gamma_i \in \bF_{q^d}\right\}\\
&=&\left\{\sum \gamma_i \alpha_i~\vert~\forall i, \gamma_i \in \bF_{q^d}\right\}=\sum\alpha_i\bF_{q^d}=V,
\end{eqnarray*}
and hence $V\in\bC_d$.
\end{proof}

\end{document}